\documentclass[12pt]{article}

\usepackage[colorlinks=false]{hyperref}
\usepackage{amssymb,amsthm}
\usepackage{amsmath,amsfonts}
\usepackage{xcolor}
\usepackage[latin1]{inputenc}

\makeatletter
\def\@seccntformat#1{\csname the#1\endcsname.\ } % the column after a section number
\makeatother

\newtheorem{lemma}{Lemma}
\newtheorem{theorem}{Theorem}
\newtheorem{proposition}{Proposition}
\newtheorem{conjecture}{Conjecture}
\newtheorem{corollary}{Corollary}
\theoremstyle{remark}
\newtheorem{remark}{\bf Remark}
\newtheorem{example}{Example}
\theoremstyle{definition}
\newtheorem{definition}{\bf Definition}

\newcommand{\bbF}{\mathbb{F}}
\newcommand{\J}{\mathcal{J}}

\newcommand{\Z}{\mathbb{Z}}

\newcommand{\V}{\mathcal{V}}
\newcommand{\C}{\mathcal{C}}
\newcommand{\A}{\mathcal{A}}
\newcommand{\Aa}{\mathcal{A}'_q}
\newcommand{\Ab}{\mathcal{A}''_q}
\newcommand{\N}{N}%
\newcommand{\vc}[1]{{\bar{#1}}}

\DeclareMathOperator{\wt}{wt}
\DeclareMathOperator{\supp}{supp}

\newcommand\param[5][\cdot]{(#2,#1,#4;#3)_{#5}}
\title{A family of diameter perfect constant-weight codes from Steiner systems%
\thanks{This is the accepted version of the paper published in Journal of Combinatorial Theory, Series~A, 200 (2023), 105790, \url{https://doi.org/10.1016/j.jcta.2023.105790}.
\\
This research is supported in part by National Natural Science Foundation of China (12071001);
the work
of D.\,S.\,Krotov is supported within the framework of the state contract of the Sobolev Institute of
Mathematics (Project FWNF-2022-0017).}}

\author{Minjia~Shi%
\thanks{Key Laboratory of Intel ligent Computing Signal Processing, Ministry of
Education, School of Mathematical Sciences, Anhui University, Hefei 230601,
China; State Key Laboratory of integrated Service Networks, Xidian University, Xian,
710071, China}%
,
Yuhong~Xia%
\thanks{School of Mathematical Sciences, Anhui
University, Hefei 230601, China.}%
,
Denis~S.~Krotov%
\thanks{Sobolev Institute of Mathematics, Novosibirsk 630090, Russia.}
}
\date{}
\begin{document}
\maketitle
\begin{abstract}
If $S$ is a transitive metric space, then $|C|\cdot|A| \le |S|$ for any distance-$d$ code $C$ and a set $A$, ``anticode'', of diameter less than $d$. For every Steiner S$(t,k,n)$ system $S$, we show the existence of a $q$-ary constant-weight code $C$ of length~$n$, weight~$k$ (or $n-k$), and distance $d=2k-t+1$ (respectively, $d=n-t+1$) and an anticode $A$ of diameter $d-1$ such that the pair $(C,A)$ attains the code--anticode bound and the supports of the codewords of $C$ are the blocks of $S$ (respectively, the complements of the blocks of $S$). We study the problem of estimating the minimum value of $q$ for which such a code exists, and find that minimum for small values of $t$.
\end{abstract}


\section{Introduction}

Diameter perfect codes are natural generalizations of perfect codes,
fascinating structures in coding theory.
They can be defined for different metrics, especially, for the Hamming metric.
In this paper, we discuss diameter perfect codes in the subspaces of
the Hamming metric space restricted by the words whose weight is constant.

The concept of a diameter perfect codes~\cite{AhlAydKha} is based on the
the code--anticode bound,
which is a generalization of the sphere-packing bound;
the most famous referenced result with that bound
is Delsarte's~\cite[Theorem~3.9]{Delsarte:1973}.
However, the main importance of \cite[Theorem~3.9]{Delsarte:1973}
is that the code--anticode bound is proved there
for any distance-regular graph
(more generally for an association scheme), without any assumption
on its symmetries.
For symmetric metric spaces, such as Hamming and Johnson schemes,
the bound has a much simpler proof, and probably was known earlier
(in particular, the key lemma in the proof of
the Elias--Bassalygo bound~\cite{Bassalygo:65}
is based on the same generalized pigeonhole principle).

\begin{lemma}[code--anticode bound]\label{L1}
Let $\V$ be a finite metric space
whose automorphism group $\mathrm{Iso}(\V)$ acts transitively on its points.
Let $\C$ be a subset of $\V$ (a \emph{code})
with minimum distance $d$ between elements,
and let $\A$ be a subset of $\V$ of diameter $D$
(an \emph{anticode}), where $D<d$. Then
\begin{equation}\label{eq:CAB}
|\C|\cdot |\A|\le |\V|.
\end{equation}
\end{lemma}
\begin{proof}
Denote
$\overline A = \{\pi(\A):\
\pi \in \mathrm{Iso}(\V) \}$.
Because of the transitivity, each point
of $\V$ belongs to exactly
$s$ sets in $\overline A$,
where $s = |\overline A|\cdot |\A| \,/\, |\V|$.
From the hypothesis on the diameter of
$\A$ and minimum distance of $\C$,
we see that each set from
$\overline A$ contains at most one
element of $\C$.
On the other hand, each element of $\C$
is contained in $s$ elements of
$\overline A$. Therefore,
$|\overline A|  \ge s |\C|$.
Substituting $s = |\overline A|\cdot |\A| \,/\, |\V|$, we get~\eqref{eq:CAB}.
\end{proof}

If \eqref{eq:CAB} holds with equality,
then $\C$ is called a
\emph{diameter perfect distance-$d$ code}
or \emph{$D$-diameter perfect code}.
It is easy to see that in this case
$\A$ is a
{maximum diameter-$D$ anticode}.
 The space $\J_q(n,w)$ of weight-$w$
 words of length $n$ over the
 $q$-ary alphabet, as was shown in~\cite{Etzion:2022:CW}, belongs
 to the class of spaces for which the code--anticode bound
 holds.
In~\cite{Etzion:2022:CW} (see also \cite[Ch.\,9]{Etzion:PCbook}),
Etzion classified known families
of non-binary diameter
perfect constant-weight codes
as follows (below, $q$ is the alphabet size,
$n$ is the length, $w$ is the weight and $d$
is the minimum distance):
\begin{itemize}
 \item[F1] Non-binary diameter-perfect constant-weight codes for
which $w = n$. Essentially, they are diameter-perfect codes in the Hamming space $\mathcal{H}(q-1,n)$
(perfect codes, extended perfect codes, MDS codes).
 \item[F2] Non-binary diameter perfect constant-weight codes for
which $w = n-1$. Known constructions \cite{EtzionVLint:2001}, \cite{Kro:2001:Ternary}, \cite{Kro:TernDiamPerf},
\cite{KOP:2016},
\cite{Svanstr:1999}, \cite{vLinTol:1999} are related to perfect codes
in the Hamming space~$\mathcal{H}(q-1,n-1)$, $q=2^m+1$, and their extensions, with one exception:
$q=3$, $n=w+1=6$, $d=4$~\cite{OstSva2002}.
 \item[F3] Generalized Steiner systems \cite{Etzion:97} (see also \cite{CJZ:2007:GSS} and references there), for which every codeword of weight $t$, $t<w$,
 is at distance $w-t$
 from exactly one codeword and $d=2(w-t)+1$ (see, e.g., \cite{ZhangGe:2013:4ary}, where $q=w=4$, $t=2$).
 \item[{F4\makebox[0mm][l] {,F6}}] \hspace{2em}$d=w$ or $d<w$, respectively.
 Codes of size $\binom{n}{w}(q-1)^{w-d+1}$,
 where all codewords of one of $\binom{n}{w}$
 possible supports form an unrestricted distance-$d$ MDS code of length $w$ over the $(q-1)$-ary alphabet
 $\{1,\ldots,q-1\}$.
 \item[F5] Codes of size $\binom{n}{w}$ with
 $d = w + 1$, where the supports of the codewords form the complete design
 (all $w$-subsets of $\{1,\ldots,n\}$).
\end{itemize}
Actually, family F5 also falls in the definition of F4 and F6 with $d=w+1$.
However, we keep it separate in the list because our goal is to extend it.
In this paper,
we further explore diameter perfect
constant-weight codes over non-binary
alphabets
and maximum anticodes corresponding to those codes.
We consider two families
of codes, F5' and F5'',
generalizing F5
(F5 corresponds to F5' with $t=w$ and to F5'' with $n-t=w$):
\begin{itemize}
 \item[F5'] Codes of size $\binom{n}{t}/\binom{w}{t}$
 with $d = 2w - t + 1$,
 where the supports of the codewords form a Steiner system $S(t,w,n)$,
 $w \le n/2$.
 \item[F5''] {Codes of size $\binom{n}{t}/\binom{n-w}{t}$ with
 $d = n - t + 1$, where the complements of the supports of the codewords
 form a Steiner system $S(t,n-w,n)$,
 $w \ge n/2$.}
\end{itemize}

It is notable that F5' and F5'' are related to the two known
families of binary diameter perfect constant-weight codes,
constructed from Steiner systems \cite[Sect.\,III]{Etzion:2022:CW}.
While in the binary case, the codes from the two families
are connected by a trivial operation (swapping the role
of $0$ and $1$), the codes from F5' and F5'' are essentially
different, except for the special case $w=n/2$.

It is worth noting that
diameter perfect codes form a proper subclass of
of the class of cardinality-optimal codes
(that is, codes with the maximum cardinality among all codes with the
same distance in the same space).
There are many constructions of optimal non-binary constant-weight codes,
most of which are not diameter perfect,
see, e.g.,
\cite{CDLL:2008:w3},
\cite{CDLL:2010:lin},
\cite{CGKLZZ:2019},
\cite{CGZZ:2015:Hanani},
\cite{CheeLing:2007:CW},
\cite{Dau:2008:MS},
\cite{Ge:2008},
\cite{KSZ:2023:CW3ary},
\cite{OstSva2002},
\cite{ZhangGe:2010},
\cite{ZZG:2012:3ary},
\cite{ZhangGe:2013:4ary},
\cite{ZhuGe:2012:*GDD}.
However, in the most of cases, a small alphabet or a small code
distance is considered,
while the codes we focus on have the minimum distance larger than the
weight and a relatively large alphabet size.
Intersections of the families of codes known before with family F5' (including F5) are mentioned in Section~\ref{s:23}
(Remarks~\ref{r:t2}, \ref{r:t3}, \ref{r:34}).

% In \cite{CDLL:2008:w3}, all parameters of optimal constant-weight codes
% of weight~$3$ and distance~$4$ are determined, including the parameters of diameter perfect codes of type F5 with $w=3$.

The structure of the paper is as follows.
In Section~\ref{section:A},
we give main definitions and basic results.
In Section~\ref{s:anticode},
we construct a class of anticodes
corresponding to families F5', F5'' and
show that they are maximum for large~$q$ (Theorems~\ref{th:anti'} and~\ref{th:anti''}),
also showing the same for a more general class
of anticodes (Proposition~\ref{max t}).
In Section~\ref{s:q0},
we observe that the existence of a Steiner
system (with the corresponding parameters)
implies the existence, for sufficiently
large~$q$,
of diameter perfect codes
that attain the code--anticode bound with the new
anticodes (Theorem~\ref{t1}), and define $q'_0$ and $q''_0$
as
the smallest possible $q$ (corresponding to families F5' and F5'', respectively) for which such codes exist.

In Section~\ref{s:23},
we find $q'_0$
in the case
of Steiner systems of strength~$2$
(Theorems~\ref{t2}) and~$3$
(Theorems~\ref{t3}, Corollaries~\ref{c:nw3} and~\ref{c:4mod12}).
Further, in Sections~\ref{s:t}
and~\ref{s:special},
we derive an upper bound
on $q'_0$ in the general case
(Theorems~\ref{t4})
and consider special cases related to the
Steiner systems $S(4,5,11)$
(Proposition~\ref{p:11,5,4}),
$S(3,4,14)$ and $S(3,4,20)$ (Proposition~\ref{p:n,4,3}).
In Section~\ref{s'},
we consider $q''_0$;
we derive a lower bound
(Proposition~\ref{q''}, Corollary~\ref{p:2''})
and solve the special cases related
to affine and projective planes
(Corollaries~\ref{c:affine} and ~\ref{c:proj}).
Section $5$ contains concluding remarks and conjectures.

% {\color{green}
% In~\cite{Etzion:2022:CW}, Etzion presented $12$ open problems, where in Problem $12$, we find a new maximum size
% of these $t$-intersecting families for parameters $d>w + 1$.
% }

\section{Preliminaries}\label{section:A}

In this section, we present definitions and notations, as well as several results to facilitate understanding throughout the paper.

By $\Z_q$,
we denote the set $\{0,\ldots, q-1\}$
(we do not need to endow it by any algebraic operation, such as addition or multiplication);
we will also use the notation
 $[n]$ for $\{1,\ldots,n\}$.
Given a positive integer $n$,
let $\Z_q^n$ be the set of all words of length~$n$ over the alphabet $\Z_q$; for each word $\vc x\in \Z_q^n$, we have $\vc x=(x_i)_{i\in [n]}$
and $x_i\in \Z_q$ for any $i\in [n]$.
The \emph{support} of a word $\vc x\in \Z_q^n$
is the set $\supp(\vc x)=\{i\in [n]:\ x_i \neq 0\}$.
The \emph{weight} of a word $\vc x\in \Z_q^n$, denoted by $\wt(\vc x)$, is equal to $|\supp(\vc x)|$.
The \emph{distance} (Hamming distance) $d(\vc x,\vc y)$ between two words
$\vc x$, $\vc y$ in $\Z_q^n$ is defined to be the number
of positions in which $\vc x$ and $\vc y$ differ.
By $\J_q(n,w)$,
we denote
the subset of $\Z_q^n$
that is restricted
by words of weight $w$,
and the corresponding
metric subspace
of the Hamming space.

\begin{definition}[codes]
Given a metric space, a \emph{code}
as a subset with at least two vertices, referred to as \emph{codewords}.
A code $\C$ is said to have \emph{distance~$d$}
if the distance between any two distinct codewords is not less than $d$.
A code $\C$ is a \emph{$q$-ary code of length $n$}
if $\C\subseteq \Z_q^n$.
If, additionally,
$\C \subseteq \J_q(n,w)$,
then it is a \emph{constant-weight},
or \emph{constant-weight-$w$}, code.
For brevity, we call
a $q$-ary constant-weight-$w$ code
of distance~$d$ an
$\param[|\C|]{n}{w}{d}{q}$ code,
or
an $\param{n}{w}{d}{q}$ code
if the cardinality
is not specified or important.
\end{definition}

\begin{definition}[anticodes and diameter perfect codes]
An \emph{anticode} of diameter $D$
in a metric space $\V$
is a nonempty subset $\A$ of $\V$
such that
the distance between any $\vc x$ and $\vc y$ in $\A$ is not greater than $D$.
A distance-$d$ code $\C$ in $\V$
is called \emph{diameter perfect}
if it attains the code--anticode bound~\eqref{eq:CAB}
for some anticode~$\A$ of diameter smaller
than~$d$.
\end{definition}

\begin{definition}[code matrix]
For a code $\C$ in $\Z_q^n$,
we define the \emph{code matrix} whose rows
are the codewords of $\C$.
To be explicit, we can require that the rows are ordered lexicographically;
however, we will never use this order in our considerations.
\end{definition}

\begin{definition}[Steiner systems]
A (Steiner) system $S(t,k,n)$,
$0<t\le k<n$,
is a pair $S = (\N , B)$, where $\N$ is an $n$-set (whose elements are often called \emph{points}) and $B$ is a set of $k$-subsets (called \emph{blocks}) of~$\N$, 
where each $t$-subset of~$\N$ is contained in exactly one block of~$B$.
Systems $S(2, 3, n)$ are also denoted STS$(n)$ and called \emph{Steiner triple systems}, \emph{STS}.
Systems $S(3, 4, n)$ are also denoted SQS$(n)$ and called \emph{Steiner quadruple systems}, \emph{SQS}.
\end{definition}

\begin{definition}[derived systems]
For a Steiner $S(t,k,n)$ system $S = ( \N , B)$
and a set of points
$\alpha \subset \N $,
$0<|\alpha|<t $,
the system
$( \N \backslash \alpha , \{ \beta \backslash \alpha:\  \alpha \subset \beta \in B \} )$
is a Steiner $S(t-|\alpha|, k-|\alpha|, n-|\alpha|)$ system, called \emph{derived} from $S$.
\end{definition}

\begin{definition}[resolvable systems]
A Steiner $S(t,k,n)$ system $S = ( \N , B)$
is called \emph{$l$-resolvable}, $1\le l < t$, (if $l=1$, just  \emph{resolvable})
if the set of blocks~$B$ can be partitioned into subsets each of which
forms an $S(l, k, n)$ on the same
point set $\N$.
If $l=1$, those subsets are referred to as \emph{parallel classes}.\end{definition}

\begin{definition}[minimum-distance graph and chromatic number]
The\linebreak[4] \emph{minimum-distance graph}
% (in some literature, \emph{line graph})
of an $S(t,k,n)$ is a graph whose vertices
are the blocks of the system,
with two vertices joined by an edge if and only if
the corresponding blocks
intersect in $t-1$ points.
For a Steiner system $S=(\N,B)$, we denote by $\chi(S)$
the chromatic number of its minimum-distance graph
(i.e. the minimum number of cells
in a partition of $B$ such that
two distinct blocks from the same cell intersect in at most $t-2$ points).
\end{definition}

In the rest of this section,
we list some known facts in the form of lemmas,
to be utilized in the proofs of our results.

\begin{lemma}[see, e.g.,
\cite{ColMat:Steiner}] \label{l2}
The number of blocks in a Steiner system $S(t,k,n)$ is
${\binom{n}{t}}/{\binom{k}{t}}$
and every point in $\N$ is contained in ${\binom{n-1}{t-1}}/{\binom{k-1}{t-1}}$ blocks.
\end{lemma}

\begin{lemma}[{\cite[Sect.\,5.4]{ColMat:Steiner}}]
\label{l:Pi}
For a Steiner system $S(t,k,n)$,
the number $\Lambda_i(t,k,n)$ of blocks that intersect a fixed
block $B$ in exactly $i$
elements does not depend on the choice of
$B$ and equals
$\binom{k}{i} \lambda_{i,k-i} $
where $\lambda_{i,j}$, $0  \le i + j \le k$,
are recursively defined  as follows:
$$
\begin{array}{rcll}
  \lambda_{i,0} &=&
    \binom{n-i}{t-i} \big/ \binom{k-i}{t-i},
    &
    0 \le i \le t,
    \\
\lambda_{i,0} &=&	1,
    &
    t < i \le k,\\
\lambda_{i,j} &=&
\lambda_{i,j-1}-\lambda_{i+1,j-1},
    &
    0  \le i< i + j \le k.
\end{array}
$$
\end{lemma}

% In one of our proofs,
% we will use the following
% well-known fact.
\begin{lemma}\label{l:KB}
 Assume positive integers $m$ and $s$ are given,
 and we consider representations of
 $m$ as the sum of $s$ nonnegative integers:
 $m=m_1+\ldots+m_s$.
 The value $\sum_i\binom{m_i}{2}$ reaches its minimum
 if and only if $m_1,\ldots,m_s \in \{a,a+1\}$, where
 $a=\lfloor m/s \rfloor$.
\end{lemma}
\begin{proof}
Otherwise, $m_l\le m_j+2$ for some $l$, $j$.
In this case,
obviously,
$\binom{m_l}{2}+\binom{m_j}{2}>\binom{m_l+1}{2}+\binom{m_j-1}{2}$,
and $\sum\binom{m_i}{2}$ is not as small as possible.
\end{proof}

\section{Anticodes}\label{s:anticode}
Generalizing the anticode $\A^{\mathrm m}( n, w, w)$
from~\cite[IV.DE]{Etzion:2022:CW},
we define these two sets:
\begin{multline*}
 \Aa( n, w, t) =
 \{(a_{1},\ldots,a_{t},a_{t+1},\ldots,a_{n})\in \Z_{q}^n: \\
 \wt((a_{1},\ldots,a_{t}))=t,\ \wt((a_{t+1},\ldots,a_{n}))=w-t \},
\end{multline*}
where $t \le w$ and $2w -t \le n$, and
\begin{multline*}
 \Ab( n, w, t) =
 \{(0,\ldots,0,a_{t+1},\ldots,a_{n})\in \Z_{q}^n:\
 \wt((a_{t+1},\ldots,a_{n}))=w \},
\end{multline*}
where $t \le n - w$ and $2w +t \ge n$.
The main theorem in this section
establishes that
$\Aa(n,w,t)$ and $\Ab(n,w,t)$
are maximum anticodes in $\J_q(n,w)$,
assuming some conditions on the parameters hold.
\begin{theorem}\label{th:anti'}
The set $\Aa( n, w, t)$,
$t\le 2w-t \le n$, $q\ge 3$,
satisfies the following:
\begin{enumerate}
 \item[\rm(1')] $\Aa( n, w, t)$ has $\binom{n-t}{w-t}(q-1)^{w}$ words;
 \item[\rm(2')] $\Aa( n, w, t)$
 is a diameter-$(2w-t)$ anticode;
 \item[\rm(3')] if $n\ge(w-t+1)(t+1)$
 and $q$ is large enough,
 then $\Aa( n, w, t)$ is a maximum
 diameter-$(2w-t)$ anticode in $\J_q(n,w)$.
\end{enumerate}
\end{theorem}
\begin{theorem}\label{th:anti''}
The set $\Ab( n, w, t)$,
$t + w \le n \le 2w + t$, $q\ge 3$,
satisfies the following:
\begin{enumerate}
 \item[\rm(1'')] {$\Ab( n, w, t)$ is a diameter-$(n-t)$ anticode;}
 \item[\rm(2'')] {$\Ab( n, w, t)$
 has $\binom{n-t}{w}(q-1)^{w}$ words;}
 \item[\rm(3'')] {if $n\ge(n-w-t+1)(t+1)$ (i.e.,
 $n \le (w+t-1)(t+1)/t$)
 and $q$ is large enough,
 then $\Ab( n, w, t)$
 is a maximum diameter-$(n-t)$ anticode in $\J_q(n,w)$.}
\end{enumerate}
\end{theorem}

Claims (1'), (1''), (2'),
and (2'')
are obvious.
To show (3') and (3''),
we firstly prove
a more general result,
of independent interest.

\begin{proposition}\label{max t}
Let $T$ be a maximum family of different $w$-subsets of
$[n]$ such that
every two subsets from $T$ intersect in at least $s$ elements.
Let $\A_q$ be the set of all $ |T| (q-1)^w$ words over~$\Z_q$
with supports in~$T$.
\begin{enumerate}
\item[\rm(1)] If $q=2$, then $\A_2$
is a maximum diameter-$(2w-2s)$ anticode.
 \item[\rm(2)] If $q$ is large enough, then $\A_q$
               is a maximum diameter-$(2w-s)$ anticode.
 \item[\rm(3)] If $\C_q$ is a diameter perfect
$\param[|J_q(n,w)|/{|\A_q|}]{n}{w}{2w-s+1}{q}$-code,
then $\C_2$ is a diameter perfect
$\param[|J_2(n,w)|/{|\A_2|}]{n}{w}{2w-2s+2}{q}$-code,
where $\C_2$ is obtained from~$\C_q$ by replacing
all non-zero symbols by~$1$ in all codewords.
\end{enumerate}
\end{proposition}
\begin{proof}
Claim (1) is trivial and follows
from the well-known correspondence between
$s$-intersecting families of $w$-subsets of~$[n]$
and sets of diameter-$(2w-2s)$ in $J_2(n,w)$
(we only note that
for binary constant-weight codes,
it is more usual to deal with
the Johnson distance, which 
is twice smaller than the Hamming distance we consider).

To show (2), let us consider
 a diameter-$(2w-s)$ anticode~$\mathcal{B}$.

If for every two words $\vc a$, $\vc b$ in $\mathcal{B}$ it holds $|\supp(\vc a) \cap \supp(\vc b)| \ge s$,
then the words from $\mathcal{B}$ have at most $|T|$ different supports,
and thus $|\mathcal{B}|\le|\A_q|$.

Assume  $\mathcal{B}$ contains $\vc a$ and $\vc b$
such that $|\supp(\vc a)\cap \supp(\vc b)|< s$.
Denote by $B$ the set of all supports of words in $\mathcal{B}$;
we divide it into two disjoint subsets $B'$, $B''$:
 $$ B' = \{ \beta \in B:\ |\beta\cap\alpha|<s \mbox{ for some } \alpha \in B  \},$$
 $$ B'' = \{ \beta \in B:\ |\beta\cap\alpha|\ge s \mbox{ for all } \alpha \in B  \}.$$
We state two claims.
\begin{itemize}
 \item \emph{For every $\beta$ in $B'$, the number of words in $\mathcal{B}$ with support $\beta$ is smaller than $(q-1)^w - (q-2)^w $}.
 Indeed, by the definition of $B'$,
 there is a word $\vc x$ in $\mathcal{B}$  such that $|\beta\cap\supp(\vc x)|<s$.
Any word with support $\beta$ that differs with $\vc x$ in all positions of $\beta\cap\supp(\vc x)$
cannot belong to $\mathcal{B}$ because of the diameter of $\mathcal{B}$.
The number of such words is larger than $(q-2)^w $, and the claim follows.
\item \emph{It holds $|B''| \le |T|-1$}.
Indeed, by the hypothesis, $B$ contains at least one element $\alpha$ not in $|B''|$.
By the definition of $B''$, each two distinct sets in $B'' \cup \{ \alpha \}$
intersect in at least $s$ points.
Since $T$ is a maximum family with this property,
we get
$|B'' \cup \{ \alpha \}| = |B''|+1 \le |T|$.
\end{itemize}
Therefore, we have
\begin{eqnarray*}
 |\mathcal{B}| &<& |B''| (q-1)^w + |B'|\big((q-1)^w - (q-2)^w\big)
 \\
 &<& (|T|-1) (q-1)^w + \textstyle \binom{n}{w}\big((q-1)^w - (q-2)^w\big).
\end{eqnarray*}

Since the degree of the polynomial
$(q-1)^w - (q-2)^w $ is less than $w$,
we have
$|\mathcal{B}| <|T| (q-1)^w = |\A_q|$
for sufficiently large~$q$, and (2) follows.

It remains to prove (3). If $\C_q$ is a distance-$(2w-s+1)$ code,
then the union of the supports of two distinct codewords
has at least~${2w-s+1}$ points; in particular, 
those supports have at most $s-1$ points in common
and hence the distance between the corresponding codewords
of~$\C_2$ (defined as in~(3)) is at least 
$2w-2s+2$. It remains to note that
$|J_q(n,w)|/{|\A_q|}$ does not depend on~$q$,
so $|\C_q| \cdot |\A_q| = |J_q(n,w)|$ 
implies $|\C_2| \cdot |\A_2| = |J_2(n,w)|$.
\end{proof}

The next lemma implies that the set of supports of the words from $\Aa(n,w,t)$ or $\Ab(n,w,t)$
satisfies the hypothesis of Proposition~\ref{max t}.
\begin{lemma}\label{l:EKR}
[Erd\H{o}s--Ko--Rado theorem \cite{ErdosKoRado,Wilson:EKR:84}
for $s$-intersecting families]
\label{er} Let $0\le s\le w\le n$,
and let $T$ be a family of
$w$-subsets of $[n]$
such that every two subsets from $T$ intersect in at least $s$ elements.
\begin{itemize}
 \item[\rm(i)]
If $n\ge (w-s+1)(s+1)$, then
$|T|\le \binom{n-s}{w-s} .$
 \item[\rm(ii)]
If $2w-s\le n\le \frac{(w-s+1)(2w-s-1)}{w-s}$, then
$|T|\le \binom{2w-s}{w} .$
\end{itemize}
% Moreover, if $n> (w-t+1)(t+1)$, then this bound is achieved by a trivially $t$-intersecting system, that is by a family $T$ containing all the $w$-subsets of the set $[n]$ that contain a fixed $t$-subset from the set $[n]$.
\end{lemma}

Actually, the classic Erd\H{o}s--Ko--Rado theorem,
proved in~\cite{ErdosKoRado} for $n\ge s+(w-s)\binom{w}{s}^3$
and in~\cite{Wilson:EKR:84} for $n\ge (w-s+1)(s+1)$,
does not state~(ii),
which is a direct corollary of~(i), obtained by
replacing the sets from $T$ by their complements
with respect to $[n]$.
All possible maximum $s$-intersecting families
of $w$-subsets of $[n]$
are described in~\cite{AhlKha:97:intersection}.
Note that each of them, by
Proposition~\ref{max t},
produces a sequence
of maximum constant-weight
anticodes over nonbinary alphabet.
We focus on only two types of such anticodes,
$\Aa(n,w,t)$ and $\Ab(n,w,t)$,
because corresponding diameter perfect codes
can be constructed from Steiner systems $S(t,k,n)$, 
which correspond to the two known classes
of binary diameter perfect constant-weight codes 
(with respect to the anticodes $\A'_2(n,k,t)$ 
and $\A''_2(n,n-k,t)$). 
For the other types of maximum $s$-intersecting families,
it is conjectured~\cite[Conjecture~1]{Etzion:2022:CW}
that corresponding binary 
diameter perfect
constant-weight codes do not exist,
which implies (if the conjecture is true)
nonexistence of corresponding non-binary codes, 
by Proposition~\ref{max t}(3).

% {\color{red} TO ADDRESS: Reviewer: ``5. ... Please explain why diameter perfect codes cannot be constructed
% for other types of anticodes.''
% REPLY:
% We do not claim that they cannot be constructed in the other cases. We
% only claimed that they can be constructed in our two cases. We have
% rephrased this place, not to confuse the reader anymore.
% TODO: rephrase, maybe some reference to Delsarte's conjecture}

\begin{proof}[Proof of Theorem~\ref{th:anti'}(3') and Theorem~\ref{th:anti''}(3'')]
From Lemma~\ref{l:EKR}(i) we see that the set
of supports from $\Aa(n,w,t)$,
under the conditions of Theorem~\ref{th:anti'},
forms a maximum $t$-intersecting family;
hence, (3') follows from Proposition~\ref{max t}.

Further, denote $s=2w-n+t$.
We see that
the set
of supports from $\Ab(n,w,t)$
forms an $s$-intersecting family.
The hypothesis
$n\le \frac{(w-s+1)(2w-s-1)}{w-s}$
of Lemma~\ref{l:EKR}(ii)
is equivalent to the condition
$n\ge (n-w-t+1)(t+1)$ of Theorem~\ref{th:anti''},
while $2w-s\le n$ is equivalent
to the trivial $0\le t$.
Therefore, we can apply  Lemma~\ref{l:EKR}(ii)
and conclude that we have a
maximum $s$-intersecting family.
Proposition~\ref{max t} completes the proof of (3'').
\end{proof}

\section[Codes and bounds for the smallest q]{Codes and bounds for the smallest $q$}\label{s:q0}

\begin{theorem}\label{t1}
Let $t$, $k$, and $n$ be integers
such that $0<t\le k<n$.
If there exists a Steiner system $S(t,k,n)$,
then the following assertions hold.
\begin{itemize}
 \item[\rm(i)] For $w=k$ and
  $q \ge \frac{\binom{n-1}{t-1}}{\binom{w-1}{t-1}}+1$, there is an
$\param[{\binom{n}{t}}/{\binom{w}{t}}]{n}{w}{2w-t+1}{q}$ code,
which attains the code--anticode bound
with the anticode $\Aa(n,w,t)$.

 \item[\rm(ii)] For $w=n-k$ and
  $q \ge \frac{\binom{n}{t}}{\binom{k}{t}} - \frac{\binom{n-1}{t-1}}{\binom{k-1}{t-1}}+1$,
 there is an
$\param[{\binom{n}{t}}/{\binom{n-w}{t}}]{n}{w}{n-t+1}{q}$ code,
which attains the code--anticode bound
with the anticode $\Ab(n,w,t)$.
\end{itemize}
If there is no $S(t,k,n)$,
then $\param[{\binom{n}{t}}/{\binom{k}{t}}]{n}{k}{2k-t+1}{q}$ codes
and
$\param[{\binom{n}{t}}/{\binom{k}{t}}]{n}{n-k}{n-t+1}{q}$ codes do not exist for any $q$.
\end{theorem}

\begin{proof}
(i)
Assume there exists a Steiner system $S(t,w,n)$
with point set $[n]$ and
blocks $\Delta_{1}$, $\Delta_{2}$, \ldots, $\Delta_M$,
$M={\binom{n}{t}}/{\binom{w}{t}}$.
Let $\vc x^i$, $i=1,\ldots,M$, be the length-$n$ word containing $0$ in the $j$th position
if $j\not\in \Delta_i$ and containing a positive value $a$ 
in the $j$th position
if $\Delta_i$ is the $a$th block with $j$
(that is, if $j\in \Delta_i$ and $a=|\{ l \in \{1,\ldots ,i\}:\ j\in \Delta_l\}|$).
According to this definition, $\wt(\vc x^i)=w$
and two different words
$\vc x^i$, $\vc x^l$ coincide only in positions where they both have zeros.
By the definition of Steiner system,
the supports of such $\vc x^i$ and $\vc x^l$ have no more than $t-1$ common elements;
hence their union has size at least $2w-(t-1)$, and $d(\vc x^i,\vc x^l) \ge 2w-t+1$.
It remains to note that by Lemma~\ref{l2} all codewords
involve $\binom{n-1}{t-1}/\binom{w-1}{t-1}$ nonzero symbols.
We conclude that $\{\vc x^1,\ldots,\vc x^M\}$  can be treated as an
$\param[{\binom{n}{t}}/{\binom{w}{t}}]{n}{w}{2w-t+1}{q}$ code
for every $q \ge \binom{n-1}{t-1}/\binom{w-1}{t-1}+1$.

(ii)
Similarly, assume there exists a Steiner system $S(t,n-w,n)$
with point set $[n]$ and
blocks $\Delta_{1}$, $\Delta_{2}$, \ldots, $\Delta_M$,
$M={\binom{n}{t}}/{\binom{n-w}{t}}$.
Similarly to~(i), we can construct
a code such that
the codeword supports are the complements of the Steiner system blocks
and 
for each coordinate all the nonzero elements
of the codewords in that coordinate are distinct.
All codewords have weight~$w$,
and any two different codewords
$\vc x$, $\vc y$ coincide only in positions where they both have zeros.
Since such $\vc x$ and $\vc y$ have no more than $t-1$ common zero positions,
their union has size at least $n-(t-1)$, and $d(\vc x,\vc y) \ge n-t+1$.
It remains to note that by Lemma~\ref{l2} all codewords
involve $\binom{n}{t}/\binom{w}{t}-\binom{n-1}{t-1}/\binom{w-1}{t-1}$
nonzero symbols.
We conclude that $\{\vc x^1,\ldots,\vc x^M\}$  can be treated as an
$\param[{\binom{n}{t}}/{\binom{n-w}{t}}]{n}{w}{n-t+1}{q}$ code
for every $q \ge \binom{n}{t}/\binom{w}{t}-\binom{n-1}{t-1}/\binom{w-1}{t-1}+1$.

It remains to show the necessity of a Steiner system.
It follows from Proposition~\ref{max t}(3), 
because the supports of the codewords
of an
$\param[{\binom{n}{t}}/{\binom{k}{t}}]{n}{k}{2k-2t+2}{2}$ code,
as well as
the complements of the codeword supports 
of an
$\param[{\binom{n}{t}}/{\binom{k}{t}}]{n}{n-k}{2k-2t+2}{2}$
code ($w=n-k$ and $s=n-2k+t$ in Proposition~\ref{max t}(3)),
form an $S(t,k,n)$.
\end{proof}

\begin{remark}\label{r:w=t}
The special cases $w=t$ and $n-w=t$ of Theorem~\ref{t1} (cases (i) and (ii), respectively)
correspond to \cite[Theorem~16]{Etzion:2022:CW},
where the code distance is $w+1$.
If $w>t$ and $n-w>t$, then the code distances
in Theorem~\ref{t1} are larger than $ w +1$.
\end{remark}

The bounds for $q$ given in Theorem~\ref{t1}
are not best possible. For each $t$, $k$,
and $n$ such that there is
a Steiner system $S(t,k,n)$,
we denote
\begin{itemize}
 \item
by $q'_{0}(t,k,n)$ the smallest $q$
such that
an
$\param[{\binom{n}{t}}/{\binom{k}{t}}]{n}{k}{2k-t+1}{q}$ code exists
 \item
by $q''_{0}(t,k,n)$ the smallest $q$
such that
an
$\param[{\binom{n}{t}}/{\binom{k}{t}}]{n}{n-k}{n-t+1}{q}$ code exists
\end{itemize}
We emphasize that in the notation $q'_{0}(t,k,n)$ the parameter $k$
coincides with the weight of the corresponding constant-weight codes;
that is why we will often see $w$ at the place of that parameter.
The situation is different with $q''_{0}(t,k,n)$, where
the corresponding weight is $n-k$.
From Theorem~\ref{t1}, we get the following bounds for $q'_0$ and $q''_0$,
which will be called \emph{trivial} and can be further improved.
\begin{corollary}
  \label{c1}
$q'_0(t,w,n) \le \frac{\binom{n-1}{t-1}}{\binom{w-1}{t-1}}+1$;\quad
$q''_0(t,k,n) \le \frac{\binom{n}{t}}{\binom{k}{t}}-\frac{\binom{n-1}{t-1}}{\binom{k-1}{t-1}}+1$.
\end{corollary}

Below, we find some values
of $q'_0$ and $q''_0$ for small $t$,
improve the trivial bound
for arbitrary $t$,
and consider in details
some concrete parameter
collections $(t,w,n)$.

\subsection[Bounds for q'0: t=2 and t=3]
{Bounds for $q'_0$: $t=2$ and $t=3$}\label{s:23}

\begin{theorem}\label{t2}
If there exists an $S(2,w,n)$, then $q'_0(2,w,n) = \frac{n-1}{w-1}+1$.
\end{theorem}

\begin{proof}
According to Corollary~\ref{c1}, $q'_0(2,w,n) \le \frac{n-1}{w-1}+1$.

On the other hand, assume we have
an $\param[\frac{n(n-1)}{w(w-1)}]{n}{w}{2w-1}{q}$ code $\C$
constructed from $S(2, w, n)$.
Take any two codewords
$\vc x$, $\vc y$ in $\C$.
By the code parameters,
$\wt(\vc x)=\wt(\vc y)=w$ and
$d(\vc x,\vc y) \ge 2w-1$.
Let
$$k=|\{i \in [n] :\, x_{i}=y_{i}\ne 0 \}|.$$
We have $2w-1\le d(\vc x,\vc y) \le 2(w-k)$;
hence $k=0$.
By Lemma~\ref{l2}, at any position, there are $\frac{n-1}{w-1}$ codewords in $\C$ that have a non-zero symbol at this position.
From $k=0$, we see that all those
non-zero symbols are distinct.
Therefore,
$q'_0(2, w, n) \ge \frac{n-1}{w-1}+1$.
\end{proof}

\begin{remark}\label{r:t2}
For $t=w=2$,
we have $q'_0(2,2,n)=n$;
this case was considered
in~\cite[Corollary~11]{Etzion:2022:CW},
see also~\cite[Theorem~9]{CheeLing:2007:CW}, where
the maximum cardinality of an $\param{n}{2}{3}{q}$ code
was determined for any~$q$ and~$n$.
Diameter perfect codes with $t=2$ and $w=3$
occur in~\cite[Theorem~6.6]{CGZZ:2015:Hanani}
as a part of the classification of optimal $\param{n}{3}{5}{q}$ codes.
\end{remark}

\begin{example}
 From the Steiner system $S(2, 3, 7),$ with the blocks
 $$\{1,2,3\},\ \{1,4,5\},\ \{1,6,7\},
\ \{2,4,6\},\ \{2,5,7\},\ \{3,4,7\},\ \{3,5,6\},$$
we can construct the following $\param[7]{7}{3}{5}{q}$ code,
$q\ge 4$:
$$\{\,
 1  1  1  0  0  0  0, \
 2  0  0  1  1  0  0, \
 3  0  0  0  0  1  1, \
 0  2  0  2  0  2  0, \
 0  3  0  0  2  0  2, \
 0  0  2  3  0  0  3, \
 0  0  3  0  3  3  0  \,\}.
 $$
The anticode $\Aa( 7, 3, 2)$ has diameter $4$ and  size $5(q-1)^{3}$.
It follows from the existence of the code that
$\Aa( 7, 3, 2)$
is a maximum diameter-$4$ anticode of length~$7$ and weight~$3$ over~$\Z_{q}$, $q\ge 4$.
\end{example}

\begin{theorem}
 \label{t3}
If there exists an $S(3,w,n)$ system~$S$, then
\begin{equation}\label{eq:mnmx}
q'_0(3, w , n) =
\min_{\text{$S$   is $S(3, w , n)$}}
\max_{\text{$D \in S'$}}
\chi(D)+1,
\end{equation}
where $S'$ is the set of all $S(2,w-1 , n-1)$ derived from~$S$.
\end{theorem}

\begin{proof}
 The upper bound
 $\displaystyle
 q'_0(3, w , n) \le
 \min_{\text{$S$ is $ S(3, w , n)$}} \max_{\text{$D \in S'$}} \chi(D)+1$
 is proved in Theorem~\ref{t4} below for more general settings.
 It remains to show the lower bound.
 Assume we have an
 $\param[\binom{n}{3}/\binom{w}{3}]{n}{w}{2w-2}{q}$
 code~$C$. From the code distance, we see that the supports
 of two different codewords cannot intersect in more than~$2$ points,
 and from the code cardinality we conclude that the codeword supports
 form an~$S(3,w,n)$.
 Next, if two different codewords $\vc{x}$ and $\vc{y}$
 have the same nonzero value in some position~$i$,
 then $d(\vc{x},\vc{y})=2w-2$ and
 the corresponding two blocks $\supp(\vc{x})$, $\supp(\vc{y})$ intersect in only one point,~$i$.
 It follows that $\supp(\vc{x})\backslash \{i\}$, $\supp(\vc{y})\backslash \{i\}$
 are disjoint blocks of the derived $S(2,w-1 , n-1)$ system~$D$
 on the point set~$[n]\backslash \{i\}$.
 Thus, if we now color the blocks of~$D$ with the value of the corresponding codeword
 in the position~$i$, the coloring will be proper for the minimum-distance graph of~$D$.
 It follows that $q \ge \chi(D)+1$.
\end{proof}

\begin{corollary}\label{c:nw3}
 If there is an $S(3, w , n)$
such that each its derived $S(2, w-1 , n-1)$ is resolvable,
then it holds
\begin{equation}\label{eq:3wn}
 q'_0(3, w , n) =
\frac{n-2}{w-2}+1.
\end{equation}
\end{corollary}
\begin{proof}
 For a resolvable $S(2, w-1 , n-1)$ system~$D$,
 we have
 $$\chi(D)=
 \frac{\binom{n-1}{2}}{\binom{w-1}{2}}\Big/
      \frac{\binom{n-1}{1}}{\binom{w-1}{1}}
 = \frac{n-2}{w-2},$$
 and \eqref{eq:mnmx} turns to~\eqref{eq:3wn}.
\end{proof}

\begin{remark}\label{r:t3}
The case $t=w=3$ was solved
in~\cite[Theorem~17]{Etzion:2022:CW}:
if $n$ is odd then $q'_0(3,3,n)=n-1$,
and $q'_0(3,3,n)=n$ if even.
The corresponding code parameters are also found in~\cite{CDLL:2008:w3},
as a part of the classification of optimal constant-weight codes with weight~$3$.
\end{remark}

Trivially, all $S(t-1, w -1, n-1)$ systems derived
from a $2$-resolvable $S(t, w , n)$ are resolvable.
It is proved in~\cite{Baker:1976} that there are  $2$-resolvable
$S(3,4,4^m)$ for each $m\ge 1$.
In~\cite{Teirlinck:94},
$2$-resolvable Steiner systems $S(3, 4, n)$
are constructed for  $n=2\cdot p^m+2$, $p\in\{7,31,127\}$, $m\ge 1$.
Moreover, from Keevash's theory~\cite{Keevash:2014}
we know that for given~$w$ and~$n$ large enough, % (i.e., $n\ge n_0(w)$),
$2$\mbox{-}resolv\-able $S(3, w,n)$ systems exist
if some divisibility conditions are satisfied.
% subject to some divisibility conditions.
In particular, for $w = 4$, those conditions
are equivalent to $n\equiv 4 \bmod 12$.
However, for
Corollary~\ref{c:nw3},
the condition $n\equiv 4 \bmod 12$
(as well as the $2$\mbox{-}divis\-i\-bil\-i\-ty
of $S(3,4,n)$)
is not necessary;
for example all $S(2,3,9)$
derived from $S(3,4,10)$
are resolvable, while
$10\not\equiv 4 \bmod 12$
and
$S(3,4,10)$ is not $2$-resolvable.

\begin{corollary}\label{c:4mod12}
If $n=4^m$ or $n=2\cdot p^m+2$, $p\in\{7,31,127\}$,
$m\in\{1,2,\ldots\}$,
or $n$ large enough satisfying
$n\equiv 4 \bmod 12$,
or $n=10$, then
$q'_0(3,4,n)=n/2$.
\end{corollary}

\begin{remark}\label{r:34}
In the case of $S(3,4,n)$,
\eqref{eq:3wn} turns to
$q_0(3,4,n)=n/2$.
In~\cite[Lemma~2.12]{ZhangGe:2010},
$2$-resolvable $S(3,4,n)$
are used to construct optimal
$\param[(q-1)n(n-1)/12]{n}{4}{6}{q}$ codes for
every $q$ in $\{3,\ldots,n/2\}$.
For $q=n/2$, this gives the same
result as in
Corollary~\ref{c:4mod12},
except the case $n=10$. However, even for $w=4$,
Corollary~\ref{c:nw3} can be more general, 
see Conjecture~\ref{conj:4mod6} in the Conclusion.
\end{remark}

\subsection[Bounds for q'0: Arbitrary t]
{Bounds for $q'_0$: arbitrary $t$}\label{s:t}

In this section, we generalize Theorem~\ref{t3}
to arbitrary $t$.
It is hard to expect that the generalization
gives a good evaluation of the real value of $q'_0$ for $t > 3$; however, it is important to conclude that
the trivial bound from Corollary~\ref{c1} is not tight.
It is also a little step towards
the research direction pointed in~\cite{Etzion:2022:CW} as Problem~4.

\begin{theorem}
\label{t4}
If there exists an $S(t,w,n)$ system $S$, $t\ge 3$, then
$$q'_0(t, w , n) \le
\frac{\binom{n-1}{t-1}}{\binom{w-1}{t-1}}-\frac{\binom{n-t+2}{2}}{\binom{w-t+2}{2}}+
\min_{\text{$S$ is $S(t, w , n)$}}
\max_{\text{$D \in S'$}}
\chi(D)+1,$$
where $S'$ is the set of all $S(2,w-t+2 , n-t+2)$ derived from~$S$.
\end{theorem}
\begin{proof}
Assume there exists an $S(t,w,n)$ system $S$ with point set $[n]$ and block set $B=\{\Delta_{1},\ldots,\Delta_{M}\}$, $M={\binom{n}{t}}/{\binom{w}{t}}$. For convenience,
we consider all operations over the point set modulo~$n$.
Denote by~$\Gamma_i$ the set $[i-t+3,i]$ of $t-2$ consequent (modulo $n$) points,
by~$S_i$ the $S(2,w-t+2,n-t+2)$ system on the point set
$[n]\backslash \Gamma_i$
derived from~$S$, and
by~$\lambda_i$ its chromatic number~$\chi(S_i)$, $i=1,\ldots,n$.
Next, we denote by~$B_i$ the subset of~$B$
consisting of all blocks
that include~$\Gamma_i$.
By the definition of~$\chi(S_i)$,
we can partition $B_i$ into $\lambda_i$ disjoint subsets,
\emph{cells}:
$$B_i=B_i^1\cup B_i^2 \cup \ldots \cup B_i^{\lambda_i}$$
such that the intersection of any two blocks
from the same cell is
$\Gamma_i$.
Let $\vc x^j=( x^j_1,\ldots, x^j_n)$, $j=1,\ldots,M$,
be the length-$n$ word
defined by assigning the exact value
in the $i$th position, $i=1,\ldots,n$:
$$
 x^j_i =
\begin{cases}
0, & \text{if $i\notin \Delta_{j}, $}\\
s, & \text{if $i\in \Delta_{j} \in B_i^s $ for some $s\in [\lambda_i]$, }\\
\lambda_i+|\{ l \in \{1,\ldots ,j\}:\ i\in \Delta_l\notin B_i\}|, & \text{if $i\in \Delta_{j} \notin B_i. $}
\end{cases}
$$
We see that $\supp(\vc x^j)=\Delta_j$;
in particular, $\wt(\vc x^j)=w$.
Any two different codewords $\vc x^j$ and $\vc x^l$
coincide either only in positions where both have zeros,
or also in positions where they both have nonzeros.
In the first case, the distance between  $\vc x^j$ and $\vc x^l$
is at least~$2w-t+1$, as desired.
Consider the second case.
In that case, $\vc x^j$ and $\vc x^l$ both have
a positive value~$s$ in the $i$th position,
i.e., $\Delta_{j}$, $\Delta_{l}$ are in the same cell $B_i^s$,
which means $\Delta_{j}\cap \Delta_{l}=\Gamma_i$.
It is easy to see that $\vc x^j$ and $\vc x^l$
coincide only in the $i$th position and in the positions where both have zero;
thus $d(\vc x^j,\vc x^l)= 2w-t+1$.
Besides, according to the definition rule of $\vc x^j$,
we know that all codewords involve
\begin{equation}\label{eq:qqq}
 \binom{n-1}{t-1}\Big/ \binom{w-1}{t-1}-\binom{n-t+2}{2}\Big/ \binom{w-t+2}{2}+
\max_{\text{$D \in S'$}}
\chi(D)
\end{equation}
nonzero symbols,
and $\vc x^1$, \ldots, $\vc x^M$ form
an $\param[{\binom{n}{t}}/{\binom{w}{t}}]{n}{w}{2w-t+1}{q}$
code for any $q$ larger than~\eqref{eq:qqq}.
\end{proof}

\subsection[Bounds for q'0:
S(4,5,11), SQS(14), and SQS(20)]
{Bounds for $q'_0$:
S(4,5,11), SQS(14), and SQS(20)}\label{s:special}

In this section, we evaluate $q'_0(4,5,11)$,
related to a unique $S(4,5,11)$,
derived from the small Witt design $S(5,6,12)$~\cite{Witt:37},
and find $q'_0(3,4,14)$ and $q'_0(3,4,20)$.

\begin{proposition}\label{p:11,5,4}
It holds $9\le q'_0(4,5,11) \le 11$.
\end{proposition}
\begin{proof}
Computing $\chi(S')$ for the unique
(up to equivalence) $S(3,4,10)$,
which is derived from $S(4,5,11)$,
gives
$ q'_0(4,5,11) \ge \chi(S')+1 = 9 $.

An example of a $\param[66]{11}{5}{7}{11}$ code, which shows that $q'_0(4,5,11) \le 11$,
is the cyclic code
with representatives
0000a100615, 00080901087, 00103009406, 00022013005, 00070040726, 0000035480a
(each representative generates $11$ codewords by cyclic shifting).
\end{proof}

\begin{proposition}\label{p:n,4,3}
The values $q'_0(3,4,14)$ and $q'_0(3,4,20)$ equal $9$ and $11$, respectively.
\end{proposition}
\begin{proof}
There are two STS$(13)$, up to permutation of points, and each of them
has chromatic number~$8$ (a computational result). 
By Theorem~\ref{t3}, we have $q'_0(3,4,14) = 8+1$.

The number of blocks in an STS$(19)$ $S'$ is $57$, and the maximum number
of mutually disjoint blocks is $\lfloor 19/3 \rfloor = 6$.
So, $\chi(S')\ge \lceil 57/6 \rceil = 10$ and hence $q'_0(3,4,20)\ge 11$ by Theorem~\ref{t3}.

Consider the cyclic SQS$(20)$ (II.A.1) from~\cite{Phelps:80:SQS20}.
It consists of the blocks
    $\{1,6,11,16\}$,
    $\{1,3,11,13\}$,
    $\{1,5,11,15\}$,
    $\{1,4,9,17\}$,
    $\{1,4,7,14\}$,
    $\{1,2,3,12\}$,
    $\{1,5,7,10\}$,
    $\{1,3,7,8\}$,
    $\{1,3,4,10\}$,
    $\{1,6,7,9\}$,
    $\{1,3,9,14\}$,
    $\{1,2,6,10\}$,\linebreak[4]
    $\{1,2,9,18\}$,
    $\{1,2,7,13\}$,
    $\{1,2,5,14\}$,
    $\{1,4,6,8\}$
    and all their cyclic shifts
    $\{a_1,a_2,a_3,a_4\} \to \{a_1+i,a_2+i,a_3+i,a_4+i\}$,
    where $+$ is modulo $20$.
Because of the cyclicity, all derived STS are isomorphic,
and we consider the derived STS $S'$ on the point set $\{1,\ldots,19\}$.
The block set of~$S'$ is partitioned into $10$
cells with mutually disjoint blocks in each cell:

\{\{1, 2, 11\}, \{3, 5, 7\}, \{4, 10, 14\}, \{6, 17, 19\}, \{8, 12, 15\}, \{13, 16, 18\}\},

\{\{1, 3, 15\}, \{2, 10, 12\}, \{5, 13, 17\}, \{6, 11, 18\}, \{7, 16, 19\}, \{8, 9, 14\}\},

\{\{1, 4, 13\}, \{2, 15, 18\}, \{3, 12, 19\}, \{5, 6, 8\}, \{7, 14, 17\}, \{9, 10, 11\}\},

\{\{1, 5, 9\}, \{2, 8, 13\}, \{3, 14, 18\}, \{4, 7, 12\}, \{6, 10, 16\}, \{11, 15, 17\}\},

\{\{1, 6, 12\}, \{2, 4, 17\}, \{3, 8, 16\}, \{5, 10, 15\}, \{9, 18, 19\}, \{11, 13, 14\}\},

\{\{1, 7, 18\}, \{2, 5, 16\}, \{3, 10, 17\}, \{4, 8, 19\}, \{6, 14, 15\}, \{9, 12, 13\}\},

\{\{1, 8, 17\}, \{2, 14, 19\}, \{3, 6, 13\}, \{4, 5, 18\}, \{7, 9, 15\}, \{11, 12, 16\}\},

\{\{2, 6, 7\}, \{3, 4, 11\}, \{5, 12, 14\}, \{8, 10, 18\}, \{9, 16, 17\}, \{13, 15, 19\}\},

\{\{1, 10, 19\}, \{2, 3, 9\}, \{4, 15, 16\}, \{7, 8, 11\}, \{12, 17, 18\}\},

\{\{1, 14, 16\}, \{4, 6, 9\}, \{5, 11, 19\}, \{7, 10, 13\}\}.

So, $\chi(S')=10$ and $q'_0(20,4,3) = \chi(S') + 1 = 11$.
\end{proof}

\subsection[Bounds for q''0]{Bounds for $q''_0$}\label{s'}

The trivial upper bound on $q''_0(t,k,n)$ in Corollary~\ref{c1} comes from a construction
where in each column of the code matrix all nonzero symbols are distinct.
Merging some symbols results in decreasing the distance between some codewords.
However, this does not always lead to decreasing the code distance,
and in general a code with the same parameters and a smaller number of symbols used
can be found.
This is not the case if $t=1$ because the distance between any two codewords
in the trivial construction already coincides with the minimum distance of the code.

\begin{corollary}
$q''_0(1,k,n)=n/k$, i.e., an $\param[n/k]{n}{n-k}{n}{q}$ code exists if and only if
$k$ divides $n$ (necessary and sufficient condition for the existence of $S(1,k,n)$)
and $q\ge n/k$.
\end{corollary}

In the following proposition,
we estimate the number of symbol pairs that can be merged if $t>1$.
The formula there looks a bit complicate,
but as we see from the corollaries below,
for special cases it becomes much simpler.

\begin{proposition}\label{q''}
If there exists an
$\param[M]{n}{w}{n-t+1}{q}$
code $\C$, where $w=n-k$ and $M=\binom{n}{t} \big/ \binom{k}{t}$, then
\begin{itemize}
 \item [\rm(i)]
 $\C$ is a diameter perfect code in $\J_q(n,w)$;
 \item [\rm(ii)]
 the supports of the codeword
 of $\C$
are the complements
of the blocks of an $S(t,k,n)$
system $S$;
 \item [\rm(iii)] it holds
\begin{eqnarray*}
(q-1) \binom{a}{2}+ba
&\le&
\lfloor\widetilde P\rfloor,
\qquad
\text{where}
\\
\displaystyle a
&=&
\left\lfloor{R/(q-1)}\right\rfloor
,
\quad
 b
=
 R - a (q-1),
\\
R
&=&
\binom{n}{t} \big/ \binom{k}{t} - \binom{n-1}{t-1} \big/ \binom{k-1}{t-1},
\\
\displaystyle \widetilde P
&=&
 \frac{M}{2n}\sum _{i=0}^{t-2}{(t-1-i) \Lambda_{i}(t,k,n)} ,
\end{eqnarray*}
and
$\Lambda_{i}(t,k,n)=\binom{k}{i} \lambda_{i,k-i}$
is from Lemma~\ref{l:Pi};
\item [\rm(iii')] in particular,\quad
$
\displaystyle
q-1 \ge R -
 \lfloor\widetilde P\rfloor.
$
\end{itemize}
\end{proposition}
\begin{proof}
 (i) Since $\C$ meets the code--anticode bound
 with the anticode
 $\Ab(n,n-k,t)$, it is
 diameter perfect.

 (ii) The cardinality of the code
 coincides the number of blocks in
 $S(t,k,n)$.
 Each codeword has exactly $k$ zeros.
 Each two codewords
 have at most $t-1$ common position
 with a zero.
 The claim is now straightforward.

 (iii)
The code matrix is an
$M \times n$
matrix. Each row contains $k$ zeros and $n-k$ nonzeros.
Each column contains $\binom{n-1}{t-1} \big/ \binom{k-1}{t-1}$
zeros (Lemma~\ref{l2}) and
$$R=\binom{n}{t} \big/ \binom{k}{t} - \binom{n-1}{t-1} \big/ \binom{k-1}{t-1} =
 \frac{(n-1)!(k-t)!(n-k)}{(n-t)!k!} $$
 nonzeros.
 The set of zero positions of each row
 corresponds to some block of~$S$.

We will say that we have
a \emph{collision} $(a,b)$
if two
different cells $a$, $b$
in the same column contain
the same nonzero value.
Let us evaluate the maximum number
of collisions.
If two rows correspond
to  blocks with $i$
common positions, then they have
at most $t-1-i$ collisions because
the code distance is $n-t+1$.
We denote the number of pairs of
such blocks by $P_i$.
Utilizing Lemma~\ref{l:Pi} we find that
$P_i=\frac{M}2   \cdot  \Lambda_i(t,k,n)$.
%$$P_0=\frac12 \cdot M
%\cdot \frac{n-t+1}{k-t+1} \cdot\bigg( \binom{k}{t-1}-1 \bigg),$$
%$$P_1=\frac12 \cdot M
%\cdot\bigg( \binom{n-t+2}{2} \bigg/ \binom{k-t+2}{2}  - 1 \bigg) - P_0,$$
%and for $2 \le i \le t-1$, we have
%$$P_i=\frac12 \cdot M
%   \cdot \bigg( \binom{n-t+i+1}{i+1} \bigg/ \binom{k-t+i+1}{i+1} - \binom{n-t+i}{i} \bigg/ \binom{k-t+i}{i} - 1 \bigg).$$
So, the number of collisions
does not exceed $P=\sum _{i=1}^{t-1}(t-1-i)P_i$.
The average number of collisions
in one column does not exceed $\widetilde P$,
where $ \widetilde P = P/n $
and there is a column with at most
$\lfloor \widetilde P \rfloor$ collisions.

On the other hand,
the number of nonzeros in each column is
$R$,
while the number of different nonzero symbols
is $q-1$.
With those constants,
the minimum number of collisions
is possible if each nonzero symbol
occurs $a$ or $a+1$ times (see Lemma~\ref{l:KB}),
where $a = \lfloor R / (q-1) \rfloor$.
This gives the left part of the inequality in (iii). (iii') coincides with (iii)
if $a=1$ (i.e., when $\widetilde P$ is not more than the half of the number of ones
in a column of the code matrix); otherwise (iii') is weaker.
\end{proof}

\begin{corollary}\label{p:2''}
In the case $t=2$, Proposition~\ref{q''} holds
with
$$ \displaystyle R= \frac{(n-1)(n-k)}{k(k-1)} \quad
\text{and }
\displaystyle \widetilde P = \frac{(n-1)\big((n-k^2)(n-1)+k(k-1)^2\big)}{2 k^2(k-1)^2} .$$
\end{corollary}
\begin{proof}
 Substituting $t=2$, we get
 $$M=\frac{n(n-1)}{k(k-1)},\quad
 R=\frac{(n-1)(n-k)}{k(k-1)},
 \quad\text{and }
 \widetilde P
  =  \frac{M}{2n} \cdot \lambda_{0,k} .$$
  From the definition of $\lambda_{i,j}$
  in Lemma~\ref{l:Pi}, we find that for $j\ge 1$
  $$
  \lambda_{i,j}=\lambda_{i,j-1}-\lambda_{i+1,j-1}=0 \qquad \text{if }
  i\ge t;
  $$
   $$
  \lambda_{t-1,j}=\lambda_{t-1,j-1}-\lambda_{t,j-1}=\frac{n-t+1}{k-t+1}-1;
  $$
   \begin{multline*}
    \lambda_{t-2,j}=\lambda_{t-2,j-1}-\lambda_{t-1,j-1} \\ =
  \frac{(n-t+2)(n-t+1)}{(k-t+2)(k-t+1)}-
  j\cdot \frac{n-t+1}{k-t+1} - (j-1).
   \end{multline*}
   Substituting $t=2$ and $j=k$, we get
   \begin{multline*}
   \lambda_{0,k}
   =
  \frac{n(n-1)}{k(k-1)}-
   k \cdot \frac{n-1}{k-1} - (k-1)
   =
  \frac{n(n-1)-{k^2(n-1)} -k(k-1)^2}{k(k-1)}.
   \end{multline*}
   It is easy to check now that
   $ \frac{M}{2n} \lambda_{0,k} $
   is exactly the expression for $\widetilde P$
   from the claim of the corollary.
\end{proof}

The next claim is about
Steiner systems related to affine planes.
\begin{corollary}\label{c:affine}
If $k$ is a prime power, then
 $ q''_0(2,k,k^2) = k^2  - \lfloor \frac{k-1}{2}\rfloor $.
\end{corollary}
\begin{proof}
 For $n=k^2$,
 $\widetilde P$ turns to $\frac{k^2-1}{2 k}$,
 and
 $\lfloor \widetilde P \rfloor =
 \lfloor \frac{k}{2} - \frac{1}{2 k} \rfloor =
 \lfloor \frac{k-1}{2}\rfloor$.
 At the same time,
 $\frac{(n-1)(n-k)}{k(k-1)}
 = k^2-1$.
 We see that
 (iii') turns to
 $q \ge k^2 - \lfloor \frac{k-1}{2}\rfloor$.

 Next, we are going to construct a
 $\param[k^2+k]{k^2}{k^2-k}{k^2-1}{q}$ code
 such that $q= k^2  - \lfloor \frac{k-1}{2}\rfloor$
 and the supports of codewords form an
 $S(2,k,k^2)$ system.
 For the role of the system,
 we take the affine plane over the finite field $\bbF_{k}$
 of order~$k$.
 The points are the pairs from $\bbF_{k}^2$
 (to treat the points as indices,  we identify them with the integers from $[k^2]$, in an arbitrary fixed order);
 the lines (blocks of the $S(2,k,k^2)$)
 are sets of points
 $(x,y)$ satisfying one of the equations
 $x +ay = b$, where $a,b \in \bbF_w$,
 or $y=b$, where $a \in \bbF_w$.
 We want the code matrix to have
 $k^2 \cdot \lfloor \frac{k-1}{2}\rfloor$ collisions
 in total, each column to have
 $ \lfloor \frac{k-1}{2}\rfloor$ pairwise disjoint collisions,
 and each two rows to have no more than one collision. We consider two cases, depending on the parity of $k$.

 \emph{Odd $k$,
 $ \lfloor \frac{k-1}{2}\rfloor = \frac{k-1}{2}$.}
 Let $S$ be a subset of $\bbF_{k}$ of size $\frac{k-1}{2}$ such that
 $0,1 \not \in S$ 
 and 
 $|S \cap \{s,-s\}|=1$
 for every~$s$ from $\bbF_{k}^*$.

For each $x_0$, $y_0$ from $\bbF_{k} $ and $s$ from $S$,
we make a collision in  the column corresponding to the point $(x_0,y_0)$ and
 the lines
 \begin{eqnarray}\label{eq:line1}
 &&\{\,(x,y) : \, x+(y_0+s)y = x_0+(y_0+s)y_0 - 1\,\}, \\
 \label{eq:line2}
 &&\{\,(x,y) : \, x+(y_0+s)y = x_0+(y_0+s)y_0 - 1+s\,\}.
 \end{eqnarray}
We now check the required properties.
\begin{itemize}
 \item
First of all,
we see that the point  $(x_0,y_0)$ does not lie
on any of the two lines \eqref{eq:line1}, \eqref{eq:line2};
hence the code matrix has nonzeros
in the corresponding two cells and we can make a collision there.
 \item Next, each point corresponds to $|S|=\lfloor \frac{k-1}{2}\rfloor$ collisions
and the collisions are pairwise disjoint
(for different $s$, the coefficients at $y$ are different, and hence
the lines are different and the corresponding rows are different too).
So, after making all these collisions, the number of symbols
used in each column becomes $k^2 - \lfloor \frac{k-1}{2}\rfloor$,
as required.
 \item Finally, for any two parallel lines
$x+ay=b$ and $x+ay = c$,
we uniquely find $s$
from $s\in\{b-c,c-b\}$, then $y_0$ from $y_0+s=a$,
and $x_0$ from~$c$ or~$b$.
This means that the rows corresponding to
such two lines have exactly one collision.
\end{itemize}
Additionally, any two intersecting lines have no collisions,
and we conclude that the code distance is $k^2 - 1$, as required.

\emph{Even $k$, $ \lfloor \frac{k-1}{2}\rfloor = \frac k2 -1$.}
 Let $S_{x_0}$, $S_{x_0}\subset \bbF_{k}$, be defined for each $x_0$ from~$\bbF_{k}$
 in such a way that:
 \begin{itemize}
  \item[(i)] $0,1 \not \in S_{x_0}$,
  \item[(ii)] for every $s$ from  $\bbF_{k}\backslash\{0,1\}$,
               it holds $|S_{x_0} \cap \{s, s+1\}|=1$,
  \item[(iii)] if $s \in S_{x_0}$, then $s \not\in S_{x_0+s}$.
 \end{itemize}
 Let us show that such  $S_{x_0}$, $x_0\in\bbF_{k}$, exists.
 For every pair $\{s,s+1\}$, $s\in \bbF_{k}\backslash\{0,1\}$,
 we divide $\bbF_{k}$ into $k/4$ quadruples of form
 $\{ z, z+s, z+1, z+s+1 \}$.
 Then, for each such quadruple, we include
 $s$ in $S_{z}$ and $S_{z+1}$,
 but not in $S_{z+s}$ or $S_{z+s+1}$,
 and include
 $s+1$ in $S_{z+s}$ and $S_{z+s+1}$,
 but not in $S_{z}$ or $S_{z+1}$.
 This rule depends on the choice of the representative $z$
 in each quadruple, but once the representatives are chosen,
 it uniquely specifies for each $x_0$, which element of
 $\bbF_{k}\backslash\{0,1\}$ belongs to $S_{x_0}$ and
 which element does not belong. Moreover, we see that
 (ii) and (iii) are satisfied automatically 
 and we can require~(i).

Now we proceed similarly to the case of odd~$k$,
but with $S_{x_0}$ instead of~$S$.
For each $x_0$, $y_0$ from $\bbF_{k} $ and $s$ from $S_{x_0}$,
we make a collision in  the column corresponding
to the point $(x_0,y_0)$ and lines
\eqref{eq:line1}, \eqref{eq:line2}.
The properties are also checked similarly,
except the last one, verifying of which differs a bit:
\begin{itemize}
   \item For any two parallel lines
$x+ay=b$ and $x+ay = c$,
we uniquely find $s$ from $s = b-c$,
then $y_0$ from $y_0+s=a$ and finally
$x_0$ from $ x_0+(y_0+s)y_0 - 1\in\{c,b\}$ and (iii).
\end{itemize}
\end{proof}

\begin{remark}
We did not make any collisions with lines of form $y=b$.
Potentially, they can also be used to decrease the number
of used symbols in some columns, but it is not enough
to decrease the alphabet in all columns.
\end{remark}

Another special case 
$S(2,k,k^2-k+1)$ is known as \emph{projective plane} of order $k-1$.
In a projective plane,
every two blocks (lines) have a common point.
It follows that in  a column of the code matrix all nonzero
symbols are distinct (alternatively,
one can see from Corollary~\ref{p:2''} 
that $P=n \widetilde P = 0$ if $n=k^2-k+1$).
\begin{corollary} \label{c:proj}
If there exists a projective plane of order~$s$, then
 $$ q''_0(2,s+1,s^2+s+1) = s^2+1.$$
\end{corollary}

\section{Conclusion}
In this paper, we explore two new classes of
diameter perfect constant-weight nonbinary codes
and
estimated the smallest values $q'_0$ and $q''_0$
of the alphabet size for which such codes exist.
Further studying these values, for special parameters,
is related to some problems in the design theory, regarding some weak form of resolvability.
Below, we formulate those problems in the form of conjectures.
Another interesting connection of weakly resolvable designs
with optimal non-binary constant-weight codes
is shown in~\cite{BZL:2022}.

It is well known that resolvable STS$(v)$ exist if and only if 
$v\equiv 3\bmod 6$ \cite{RChWil:kirkman}.
To show that $q'_0(3,4,n)=n/2$,
(Corollary~\ref{c:nw3}) we need the following:
\begin{conjecture}\label{conj:4mod6}
 If $n \equiv 4 \bmod 6$, then there is an SQS$(n)$
 such that all its derived STS are resolvable.
\end{conjecture}
Up to now, we only know that the conjecture is true
for $n=10$ and
in the cases
when $2$-resolvable SQS$(n)$ are known to exist:
for $n=4^m$~\cite{Baker:1976};
for $n=2\cdot p^m+2$, $p\in\{7,31,127\}$~\cite{Teirlinck:94};
for sufficiently large $n$, $n\equiv 4 \bmod 12$~\cite{Keevash:2014}.

If $v \equiv 1 \bmod 6$, then an STS$(v)$ cannot be resolvable,
and the number of mutually disjoint blocks in
such a system~$S'$ does not exceed $(v-1)/3$.
So, $\chi(S') \ge \lceil\frac{v(v-1)/6}{(v-1)/3}\rceil = \frac{v+1}2$.
\begin{conjecture}\label{conj:2mod6}
 If $20 \le n \equiv 2 \bmod 6$, then there is an SQS$(n)$
 such that all its derived STS have chromatic number $n/2$,
 i.e., can be partitioned into $n/2$ partial parallel classes
 (equivalently, $q'_0(3,4,n)=n/2+1$).
\end{conjecture}

\subsection*{Declaration of competing interest}
The authors declare that they have no known competing financial interests or personal
relationships that could have appeared to influence the work reported in this paper.
\subsection*{Data availability}
No data was used for the research described in the article.

%==========================================================

% \bibliographystyle{plain}
% \bibliography{../../k}
% \end{document}

\providecommand\href[2]{#2} \providecommand\url[1]{\href{#1}{#1}}
  \def\DOI#1{{\small {DOI}:
  \href{http://dx.doi.org/#1}{#1}}}\def\DOIURL#1#2{{\small{DOI}:
  \href{http://dx.doi.org/#2}{#1}}}

\end{document}